%% file: main.tex
\documentclass[12pt, letterpaper,  onecolumn]{IEEEtran}

\usepackage[top=1.0in,bottom=1.0in,left=1.0in,right=1.0in]{geometry}
\usepackage{color}
\usepackage{epsfig}
\usepackage{bm}
\usepackage[tbtags]{amsmath}
\usepackage{amsthm}
\usepackage{amssymb,amsfonts,latexsym}
\usepackage{algorithmic}
\usepackage{euscript}
\usepackage{subfigure}
\usepackage{graphics,eepic,epic,psfrag}
\usepackage{enumerate}
\usepackage{epstopdf}        
\usepackage{hyperref}
\usepackage[numbers, sort&compress]{natbib}

\graphicspath{{./figures/}}

\input{styles/macro}

\input{symbols}

\newcounter{remarkCounter}
\newcommand{\tc}{\textcolor{magenta}}

\title{Secret Key Generation from Sparse Wireless Channels: Ergodic Capacity and Secrecy Outage}

\author{
  Tzu-Han Chou,~\IEEEmembership{Student Member,~IEEE,}
  Stark C.~Draper,~\IEEEmembership{Member,~IEEE,} \\
  and Akbar M.~Sayeed,~\IEEEmembership{Fellow,~IEEE.}
  \thanks{This work has been supported in part by the National Science Foundation
   under CAREER Grant No.~CF-0844539 and Grant No.~CNS-0627589.}
  \thanks{The authors are with the Dept.~of Electrical and Computer
  Engineering, University of Wisconsin, Madison, WI 53706.  E-mail:
  thchou.th@gmail.com, sdraper@ece.wisc.edu, akbar@engr.wisc.edu.}
}

\begin{document}
\maketitle

\begin{abstract}
This paper investigates generation of a secret key from a reciprocal
wireless channel. In particular we consider wireless channels that
exhibit sparse structure in the wideband regime and the impact of the
sparsity on the secret key capacity.  We explore this problem in two
steps. First, we study key generation from a \emph{state-dependent
  discrete memoryless multiple source}. The state of source captures
the effect of channel sparsity.  Secondly, we consider a wireless
channel model that captures channel sparsity and correlation between
the legitimate users' channel and the eavesdropper's channel. Such
dependency can significantly reduce the secret key capacity.

According to system delay requirements, two performance measures are
considered: (i) ergodic secret key capacity and (ii) outage
probability.  We show that in the wideband regime when a white
sounding sequence is adopted, a sparser channel can achieve a higher
ergodic secret key rate than a richer channel can.  For outage
performance, we show that if the users generate secret keys at a
fraction of the ergodic capacity, the outage probability will decay
exponentially in signal bandwidth. Moreover, a larger exponent is
achieved by a richer channel.

\end{abstract}


\input{introduction}

\input{system_model3}

\input{ergodic01}
\input{outage01}

\input{conclusion}

\appendix
\input{appendix}

\bibliographystyle{ieeetr}
\bibliography{ref_papers,ref_books}

\end{document}

%% file: styles/macro.tex
\newcommand{\nnum}{\nonumber}

\newcommand{\ie}{i.e.}
\newcommand{\eg}{e.g.}

\newcommand{\cf}{cf.}

\newtheorem{theorem}{Theorem}
\newtheorem{lemma}[theorem]{Lemma}

\newtheorem{corollary}[theorem]{Corollary}
\newtheorem{definition}{Definition}

\newcommand{\eqa}{\stackrel{(a)}{=}}
\newcommand{\eqb}{\stackrel{(b)}{=}}
\newcommand{\eqc}{\stackrel{(c)}{=}}

\newcommand{\lea}{\stackrel{(a)}{\le}}
\newcommand{\leb}{\stackrel{(b)}{\le}}

\newcommand{\gea}{\stackrel{(a)}{\ge}}
\newcommand{\geb}{\stackrel{(b)}{\ge}}
\newcommand{\gec}{\stackrel{(c)}{\ge}}
\newcommand{\ged}{\stackrel{(d)}{\ge}}

\DeclareMathOperator{\diag}{diag}

\newcommand{\randcn}[2]{\mathcal{CN}\left(#1,#2\right)}  
\newcommand{\randBern}[1]{\mathrm{Bern}\left(#1\right)}  
\newcommand{\randBino}[2]{\mathrm{Bino}\left(#1,#2\right)}  

\newcommand{\calK}{\mathcal{K}}

\newcommand{\calP}{\mathcal{P}}

\newcommand{\calS}{\mathcal{S}}

\newcommand{\calX}{\mathcal{X}}
\newcommand{\calY}{\mathcal{Y}}
\newcommand{\calZ}{\mathcal{Z}}

\newcommand{\bA}{\mathbf{A}}

\newcommand{\bB}{\mathbf{B}}

\newcommand{\bd}{\mathbf{d}}
\newcommand{\bD}{\mathbf{D}}
\newcommand{\be}{\mathbf{e}}

\newcommand{\bh}{\mathbf{h}}

\newcommand{\bI}{\mathbf{I}}

\newcommand{\bR}{\mathbf{R}}

\newcommand{\bS}{\mathbf{S}}

\newcommand{\bmH}{\bm{H}}

\newcommand{\bmS}{\bm{S}}

\newcommand{\bmW}{\bm{W}}

\newcommand{\bmX}{\bm{X}}

\newcommand{\bmY}{\bm{Y}}

\newcommand{\bmZ}{\bm{Z}}

\DeclareMathAlphabet{\mathbsf}{OT1}{cmss}{bx}{n}
\DeclareMathAlphabet{\mathssf}{OT1}{cmss}{m}{sl}

\newcommand{\rvP}{\mathsf{P}}

\DeclareSymbolFont{bsfletters}{OT1}{cmss}{bx}{n}
\DeclareSymbolFont{ssfletters}{OT1}{cmss}{m}{n}
\DeclareMathSymbol{\bsfGamma}{0}{bsfletters}{'000}
\DeclareMathSymbol{\ssfGamma}{0}{ssfletters}{'000}
\DeclareMathSymbol{\bsfDelta}{0}{bsfletters}{'001}
\DeclareMathSymbol{\ssfDelta}{0}{ssfletters}{'001}
\DeclareMathSymbol{\bsfTheta}{0}{bsfletters}{'002}
\DeclareMathSymbol{\ssfTheta}{0}{ssfletters}{'002}
\DeclareMathSymbol{\bsfLambda}{0}{bsfletters}{'003}
\DeclareMathSymbol{\ssfLambda}{0}{ssfletters}{'003}
\DeclareMathSymbol{\bsfXi}{0}{bsfletters}{'004}
\DeclareMathSymbol{\ssfXi}{0}{ssfletters}{'004}
\DeclareMathSymbol{\bsfPi}{0}{bsfletters}{'005}
\DeclareMathSymbol{\ssfPi}{0}{ssfletters}{'005}
\DeclareMathSymbol{\bsfSigma}{0}{bsfletters}{'006}
\DeclareMathSymbol{\ssfSigma}{0}{ssfletters}{'006}
\DeclareMathSymbol{\bsfUpsilon}{0}{bsfletters}{'007}
\DeclareMathSymbol{\ssfUpsilon}{0}{ssfletters}{'007}
\DeclareMathSymbol{\bsfPhi}{0}{bsfletters}{'010}
\DeclareMathSymbol{\ssfPhi}{0}{ssfletters}{'010}
\DeclareMathSymbol{\bsfPsi}{0}{bsfletters}{'011}
\DeclareMathSymbol{\ssfPsi}{0}{ssfletters}{'011}
\DeclareMathSymbol{\bsfOmega}{0}{bsfletters}{'012}
\DeclareMathSymbol{\ssfOmega}{0}{ssfletters}{'012}

\newcommand{\tilbh}{\tilde{\bh}}

\newcommand{\hatK}{\hat{K}}

%% file: symbols.tex
\newcommand{\Csiszar}{\,Csisz\'{a}r~}

\newcommand{\Rphi}{R_{\phi}}
\newcommand{\ns}{M}                 
\newcommand{\slen}{N_d}             
\newcommand{\maxdelay}{\tau_{\max}}
\newcommand{\maxL}{L_{\max}}             
\newcommand{\SetSP}{\bm{S}}          
\newcommand{\varh}{\nu^2}
\newcommand{\varhe}{\upsilon^2}
\newcommand{\hmain}{\bmH_{ab}}
\newcommand{\heve}{\bmH_{e}}
\newcommand{\SPmain}{\SetSP_{ab}}    
\newcommand{\nvar}{\sigma^{2}}               

\newcommand{\Cerg}{C_\mathrm{er}}   
\newcommand{\Iab}{I_{ab}}
\newcommand{\Iae}{I_{e}}
\newcommand{\Ierg}{I_\mathrm{er}}
\newcommand{\Rerg}{R_\mathrm{er}}
\newcommand{\tsf}{\lambda}                          
\newcommand{\Iinst}{I_\mathrm{s}}

\newcommand{\Reve}{R_e}
\newcommand{\Cs}{C_s}
\newcommand{\Pout}{\rvP_\mathrm{out}}

\newcommand{\commsg}{\Phi}                   

\newcommand{\snr}{\gamma}                           

\newcommand{\taum}{\tau_{\max}}



%% file: introduction.tex
\section{Introduction}
\label{sec:sparse.intro}

The fundamental limit of secret key generation from discrete
memoryless multiple source (DMMS) is developed by Ahlswede, \Csiszar
\cite{Ahlswede_Csiszar93} and Maurer \cite{maurer93}.  Their results
show that if $X,Y,Z$ (respectively observed by Alice, Bob and Eve) are
correlated with a known distribution, it is possible to generate a
secret key between Alice and Bob at a positive rate through use of a
public discussion. The resulting information rate leaked to Eve can be
made arbitrarily small.  The supremum of achievable secret key rates
is called the \emph{secret key capacity}.

Since their work, there have been many extensions to explore the secret
key capacity of more complicated models. In \cite{Khisti_isit08,
  Prabhakaran08}, users observe DMMS and also transmit information via
wiretap channel \cite{Wyner75}, but there is no access to public
channel for discussion.  The authors in \cite{ChenVinck06, Liu07,
  Khisti_skey_asym_isit09, KhistiEtal_skey_TCSI} consider a wiretap
channel influenced by a random channel state, known by one (or both)
of the legitimate users.  In such models, the random channel state can
be viewed as a kind of correlated source shared by
transmitter/receiver which also influences the transmission.

In \cite{chou_it10, chou_allerton11}, key generation from DMMS is
considered where the DMMS is excited by a deterministic source
\cite{chou_it10} or by a random source \cite{chou_allerton11}. This
sender-excited model is motivated by an application in which key
generation is based on the inherent randomness of reciprocal wireless
channel.  Consider a situation where Alice and Bob transmit a sounding
signal to each other over a reciprocal wireless channel. Due to the
channel reciprocity, Alice and Bob observe a pair of correlated
sources. The source turns out to be a good source for secret key
generation because it can be the case that (i) the source is
correlated, (ii) the source is ubiquitous since it is from wireless
channel, and (iii) it is hard to eavesdrop because the wireless channel
varies quickly in the spatial and temporal domains.  This issue has
received much attention in terms of theoretical and practical research
\cite{hessan96, wilson_Tse_Scholtz07, Ye_vtc07, YeEtal,Wallace09,
  Patwari10,Aono05,sayeed08}.  However, most of this work is subject to the
assumption that the eavesdropper channel is statistically independent
of the main channel (the channel between Alice and Bob). This is true
when the environment has rich scattering such that the correlation
between channel coefficients decreases rapidly in the spatial domain.

However, there is growing experimental evidence (\eg, \cite{Molisch05,
  GJRTT04, Yan07, Czink07}) and physical arguments (\eg,
  \cite{Sayeed06, Bajwa09, Bajwa10}) which show
that realistic wireless channels are sparse at large bandwidths.  The
effect of channel sparsity on secret key capacity is twofold: (i) it
reduces the degrees-of-freedom (DoF) of the main (Alice-Bob) channel
and (ii) it induces spatial correlation \cite{LeeTC73}, thereby
increasing Eve's ability to observe the main channel.

We revisit the key generation problem when the channel exhibits
sparsity in the wideband regime.  This channel characteristic can be
captured by a \emph{sparsity pattern} that defines the non-zero
support of the channel coefficients. Depending on the environment, the
sparsity pattern could experience fast or slow time variations.  The
channel model also captures the correlation between the main channel and
Eve's observations.  To study secret key generation in this context we
capture these characteristics by defining a \emph{state-dependent}
discrete multiple memoryless source (SD-DMMS).  We specialize this
model to the statistical characterization of sparse wireless channels
were the sparsity pattern plays a role of the channel state and, as we
discuss next, develop ergodic capacity and secrecy outage results.

In analogy to communication over a fading channel, two regimes are
studied according to the system delay constraint:
\begin{itemize}
  \item \emph{Ergodic regime (the delay tolerant regime)}: If the key
    is generated based on a large number of observations across
    multiple states, the secret key capacity is well-defined in the
    Shannon sense. We call the capacity in this case the
    \emph{ergodic} secret key capacity. The main problem is that the
    system suffers from an excessive delay.

  \item \emph{Non-ergodic regime (the delay stringent regime)}: If the
    observed source sequence is not long enough or the state changes
    slowly so that the key generation is forced to occur within a
    period of constant state, the capacity is not defined in
    general. In this case, we consider the \emph{secrecy outage
      probability} which measures the probability that the
    instantaneous state condition cannot support the key rate to
    fulfill the secrecy condition (this will be defined later).
\end{itemize}

Secrecy outage is also considered in other research regarding
state-dependent (fading) wiretap channel (\eg, \cite{Tang09,
  GungorArXiv}).  We show that when a white sounding sequence is
adopted in the wideband (low power) regime, a sparser channel can
achieve a higher secret key rate than a richer channel
can. This is analogous to capacity behavior in sparse multi-antenna
channels in \cite{Sayeed07}.
Furthermore, at each signal-to-noise ratio (SNR), there is an
adequate bandwidth that maximizes the secret key rate.  For the outage
performance, we show that the system can achieve an exponential
decaying outage probability by using an $\alpha$\emph{-backoff} scheme
($0<\alpha\leq 1$) in which secret key rate is a fraction $\alpha$ of
the ergodic capacity.  Unlike the ergodic case, now a richer channel
always has a larger exponent characterizing the decay of the outage
probability. In a similar vein as communication over a fading channel,
this demonstrates that a large number of DoF helps to smooth out the effect of the
unknown state.

The paper is organized as following. In
Section~\ref{sec:sparse.system_model} we give some definitions and
describe the system model. This includes the correlated sparse
wireless channel model, the definition of the SD-DMMS, and the one-way
discussion key generation protocol.  In
Section~\ref{sec:sparse.ergodic}, we investigate the ergodic secret
key capacity of SD-DMMS and apply this to key generation from a sparse
wireless channel.  Outage is defined in
Section~\ref{sec:sparse.outage}. We give a necessary and sufficient
condition for an outage event and explore the outage probability when
an $\alpha$-backoff scheme is used. Detailed proofs are deferred for
the Appendix.

%% file: system_model3.tex
\section{Definitions and System Model}
\label{sec:sparse.system_model}

In this paper we are motivated by key generation based on
  wireless channel that exhibits sparsity in the delay domain.  We
  first develop our model of a sparse wireless channel in Section
  \ref{sec:sparse_channel_model}.  While in earlier works on modeling
  sparse wireless channels, e.g., see \cite{Molisch05, Cotter02, Carbonelli07, Li07,
  Raghavan07, Raghavan2011}, there is only a single channel
  to model, in Section \ref{sec:sparse_channel_model} we need to model
  the main (Alice-to-Bob) channel as well as Eve's correlated
  observations of that main channel.  Following our wireless
  motivations, in Section \ref{sec:sd_dms} we develop an abstracted
  \emph{state-dependent discrete multiple source} (SD-DMMS) model.  In
  this model the ``state'' captures the effect of the slowly varying
  sparsity pattern while the key itself is extracted from the
  conditionally-generated (conditioned on the sparsity pattern)
  channel fades.  Finally, in Section \ref{sec:key_gen_protocol} the
  \emph{one-way public discussion} key generation protocol is formally
  presented.

\subsection{Sparse reciprocal wireless channel}\label{sec:sparse_channel_model}


Consider a wireless communication system with bandwidth $W$.  Say that
the channel exhibits sparsity in the delay domain\footnote{In this
  paper, we consider channel sparsity in the delay domain. It is not
  difficult to extend the result to the sparsity in either the Doppler
  or spatial domains, e.g., \cite{Sayeed06, Bajwa09, Bajwa10,
    Raghavan2011}.}  where $\maxdelay$ is the maximum delay spread of
the channel. Then $\maxL = \lceil \maxdelay W \rceil$ is the maximum
number of resolvable paths.  A sounding sequence $\bd = [d_1, d_2,
  \cdots, d_{\slen}]^T$ is transmitted over time period $T$, where
$\slen=\lceil TW \rceil$. The sounding sequence is a known sequence
with power $\bd^H\bd = P$.  We assume each two-way (Alice
$\leftrightarrows$ Bob) sounding is done within a channel coherence
period (\ie, $T_{coh} \gg 2T$).  Further multiple channel soundings
(indexed by $t$) are performed within non-overlapping coherence
periods meaning that each set of soundings are independent.

The channel outputs in sounding interval $t$ are
\begin{subequations}\label{equ:sparse_channel_output}
\begin{align}
    \bmX[t] &= \bD \hmain[t] + \bmW_1[t] \qquad \text{(Alice)} \ , \\
    \bmY[t] &= \bD \hmain[t] + \bmW_2[t] \qquad \text{(Bob)} \ ,
\end{align}
\end{subequations}
where $\hmain[t] = (H_1[t], \cdots, H_{\maxL}[t])^T$ is the
sampled (virtual) channel coefficient
\cite{Sayeed06,sayeed_bookchapter} vector, and $\bD$ is an
$N$-by-$\maxL$ Toeplitz matrix with $N = \slen +\maxL-1$:
\begin{equation*}
    \bD =
    \left[
      \begin{array}{cccc}
        d_1       & 0   &  \cdots & 0 \\
        d_2       & d_1 &  \cdots & 0 \\
        \vdots    & d_2 &  \cdots & d_1 \\
        d_{\slen} & \vdots        &         & d_2 \\
        \vdots    & d_{\slen}     & \ddots  &     \\
        0         & 0   & \ddots  & \vdots        \\
        0         & 0   & 0       & d_{\slen}     \\
      \end{array}
    \right]
    =
    \left[
      \begin{array}{cccc}
        \bd_1, \bd_2, \cdots, \bd_{\slen}
      \end{array}
    \right] .
\end{equation*}
A widely used sounding signal is a sequence whose spectrum is
  asymptotically white in $\slen$. In this case $\bD$ is a full
column-rank matrix such that\footnote{Here and in the following, we
  say $g(x^n) \doteq g$ if $g(x^n) \to g$ when $n$ is sufficiently
  large.}
\begin{equation}\label{equ:white_sequence_assumption}
    \bD^H\bD \doteq P \bI_{\maxL}
\end{equation}
when $\slen$ is sufficiently large. One such example is
$\bd=\sqrt{P}\be = \sqrt{P}(1, 0, \cdots, 0)^T$. Another such example
is pseudo-random (PN) sequence in spread spectrum system
\cite{SSHandbook}.  The noise terms $\bmW_1[t]$ and $\bmW_2[t]$ in
\eqref{equ:sparse_channel_output} are independent
$\randcn{\mathbf{0}}{\nvar_a \bI_{N}}$ and
$\randcn{\mathbf{0}}{\nvar_b \bI_{N}}$ vectors, respectively.

\subsubsection{Sparse channel model}
Most channels that have a small number of physical paths will exhibit
sparsity in the delay domain as the signal bandwidth $W$
increases. In particular, in some delay bin $\ell$, the
  corresponding channel coefficient $H_\ell[t]$ will be zero.  In
this paper, we adopt the \emph{sub-linear} law model considered in
previous work \cite{Raghavan07, Raghavan2011} to capture the sparse
channel characteristic. In this model, the channel is called
\emph{$\delta$-sparse} if the average number of non-zero channel
coefficients scales as
\begin{equation}\label{equ:sublinear_law}
    L = (\maxdelay W)^{\delta} = \maxL^{\delta}, \qquad \delta \in (0,1) \ .
\end{equation}
The parameter $L$ is also the mean number of channel DoF.

The \emph{channel sparsity pattern} of the main channel in sounding
interval $t$ is
\begin{equation*}
 \SPmain[t] = \Big( S_{ab,1}[t], \cdots, S_{ab,\maxL}[t] \Big) \in \calS^{\maxL},
\end{equation*}
where $\calS=\{0,1\}$ and $E\left[\sum_{\ell=1}^{\maxL}
  S_{ab,\ell}[t]\right]=L$.  This pattern defines the support of the
channel vector
\begin{equation*}
  \hmain[t] = \Big( H_1[t], \ H_2[t] \, \ldots \, H_{\maxL}[t] \Big),
\end{equation*}
i.e., $H_{\ell}[t] = 0$ if and only if $S_{ab, \ell}[t] = 0$.  The
channel coefficients $H_\ell[t]$ are independent
$\randcn{0}{\varh_\ell}$ variable where the variance $\varh_\ell=0$ if
$S_{ab,\ell}[t] = 0$.  The channel has \emph{unit} power, \ie,
$\sum_{\ell} \varh_\ell = 1$.  Later, we use ``channel
degrees-of-freedom'' (DoF) to refer to the \emph{weight} of the
realization of the vector $\SPmain[t]$.  We also call $\SPmain[t]$ the
\emph{state} of $\hmain[t]$.  A \emph{rich} multipath channel
corresponds to $\delta \to 1$.

The sparsity pattern $\SPmain[t]$ will, in general, be
time-varying. However, in most case of interest, $\SPmain[t]$ will
change much more slowly than the channel coefficients $\hmain[t]$.
This is because the main reflectors, by which paths are resolved by
different delay bins, move more slowly than the phase changes that
influencs the fading coefficients \cite{sayeed_bookchapter, Proakis_DigitComm, Goldsmith2005}.
Because of this, most of the secret key rate will
be generated by the randomness inherent to the channel coefficients
rather than the sparsity pattern itself.  Furthermore, there exists
good techniques to estimate the sparsity pattern reliably based on few
observations, \eg, \cite{FRG09}. Thus, we consider $\SPmain[t]$ known to
Alice and Bob.  Let $n$ be the number of channel sounding periods
during which the sparsity pattern remains constant.  We term this the
\emph{sparsity coherence period}.  Thus, the $m$-th sparsity coherent
period extends from $t = (m-1)n+1$ to $t = mn$.  In this interval
$\SPmain[t]$ remains constant, \ie, $\SPmain[t] = \SPmain[mn]$ for all
$t$, $(m-1)n+1 \leq t \leq mn$.  We further assume that $\SPmain[t]$ is
independent across periods.

Modeling the distribution of the state itself is a difficult task, so
we consider a simple model
\begin{equation}
    \Pr(S_{ab,\ell}=1) = \frac{L}{\maxL} =  (\maxdelay W)^{-(1-\delta)} \triangleq \rho
\end{equation}
for all $\ell$.  In other words, the $S_{ab,\ell}$ is Bernoulli
distribution with parameter $\rho$ (denoted $\randBern{\rho}$).

\subsubsection{Eavesdropper's correlation model}
Eve's channel output is similar to
\eqref{equ:sparse_channel_output}\footnote{In order to get meaningful
  observations, we assume Eve is located close to one of the users. So
  only one of the two Eve's channel outputs during the two-way
  sounding correlates with the main channel. The other output is
  independent of the main channel due to fast spatial decorrelation.}:
\begin{align}\label{equ:sparse_Eve_output}
    \bmZ[t] &= \bD \heve[t] + \bmW_3[t]  \qquad \text{(Eve)} \ ,
\end{align}
where the noise is $\randcn{\mathbf{0}}{\nvar_e \bI_{N}}$. The channel
coefficient vector $\heve[t]\!=\!(H_{e,1}[t], \cdots, H_{e,\maxL}[t])^T$
is also $\delta$-sparse with state denoted by $\SetSP_e[t]$, and each
element $\randcn{0}{\varhe_\ell}$ distributed. We model the
correlation between $\heve[t]$ and $\hmain[t]$ in a two-step process
as follows:
\begin{itemize}
  \item \emph{Correlation between $\SetSP_e$ and $\SPmain$}: For each
    delay bin $\ell$ for which $S_{ab,\ell}=1$, the probability that
    Eve also has non-zero channel gain is $\theta$. i.e.,
      \begin{equation}\label{equ:overlap_model}
         \Pr(S_{e,\ell}=1| S_{ab,\ell}=1) = \theta
      \end{equation}
      for all $1 \leq \ell \leq \maxL$.

  \item \emph{Correlation between individual channel coefficient}: For
    those channel coefficients in the ``common support'' delay bins,
    \ie, in the set $\{\ell: S_{ab,\ell} = S_{e,\ell} = 1\}$, the
    correlation coefficients are
      \begin{align*}
         &\eta(H_\ell, H_{e,\ell}) \triangleq \frac{E[H_\ell H^*_{e,\ell}]}{\sqrt{E[|H_\ell|^2]E[|H_{e,\ell}|^2]}} = \eta \ .
      \end{align*}
\end{itemize}

The parameter $\theta$ captures the fraction of DoF that the main
channel and Eve's channel have in common. One can think of the
relationship between the state (sparsity pattern) of the main channel
and that of the eavesdropper's observation as a binary memoryless
channel.  However, because Eve's marginal channel has the same
$\delta$-sparsity as the main channel (since the users are in the same
environment), the channel is not symmetric.  In other words,
transition probability $\Pr(S_{e,\ell}=1| S_{ab,\ell}=0) \neq
\Pr(S_{e,\ell}=0| S_{ab,\ell}=1) $.  This is illustrated in Figure
\ref{fig:basc}.  Finally, the parameter $\eta$ captures the effect
that the paths (of both channel) located in the common delay bin
shares the same physical scattering.

\begin{figure}
  \vspace{1mm}
  \center
  \includegraphics[width=0.4\textwidth]{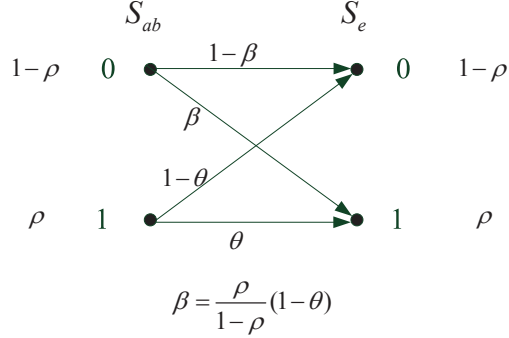}\\
  \caption{Transition probability $\Pr(S_{e,\ell}| S_{ab,\ell})$. }
  \label{fig:basc}
\end{figure}

\emph{Remark} \arabic{remarkCounter}: \stepcounter{remarkCounter} The
parameter space $\{(\theta, \eta), \delta\}$ of our model captures
many scenarios of interest.  From a physical aspect, there are two
factors effecting Eve's channel correlation: the distance between Eve
and Bob (which mainly impacts leakage to Eve), and richness/sparseness
of the multipath (which impacts both leakage to Eve and the common
randomness between Alice and Bob).  When Eve gets close to Bob,
generally, both $\theta$ and $\eta$ will increase (and vice versa);
this will generally increase the leakage.  The parameter $\delta$ (and
thus $\rho$) controls the maximum number of DoF.  When multipath is
rich, $\rho$ is high and $\eta$ is closer to zero (\ie, high overlap
but independent), resulting in the highest capacity and lowest
leakage.  For sparse multipath, $\rho$ is lower (lower common
randomness) and $\eta$ could be large even at larger distances between
Eve and Bob. In this case leakage will likely increase more slowly as
Eve gets closer to Bob.

To generate a secret key, users repeat the channel sounding
\eqref{equ:sparse_channel_output} (and \eqref{equ:sparse_Eve_output})
$n\ns$ times and generate a key based on a pair of super-block
$\{(\bmX[t], \SPmain[t]), (\bmY[t], \SPmain[t])\}_{t=1}^{n\ns}$.  In the following section we
abstract away the actual sounding process and specify a
state-dependent source model where the state varies more slowly than
the underlying source-realization process from which the key is
generated.  When we study the ergodic case, we will let both $n$ and
$\ns$ go to infinity, while when we study the outage case, $\ns=1$,
and $n$ can be large.

\subsection{State-dependent discrete memoryless multiple source}\label{sec:sd_dms}

To leverage results on information theoretic security, we consider a
state-dependent (SD) DMMS model depicted in
Figure \ref{fig:sddmms_model}. The observation triple
$(X^{nM},Y^{n\ns},Z^{n\ns}) \in \calX^{n\ns} \times \calY^{n\ns}
\times \calZ^{n\ns}$ is generated according to
$p(x^{n\ns},y^{n\ns},z^{n\ns}|s_{ab}^\ns, s_e^\ns)$, conditioning on
the pair of length-$\ns$ sequences: $(s_{ab}^\ns, s_e^\ns) \in
\calS^\ns \times \calS^\ns$.

As discussed in Section~\ref{sec:sparse_channel_model}, $S_{ab}^\ns$
is the state sequence of Alice and Bob's correlated source
$X^{n\ns}$,$Y^{n\ns}$ and $S_e^\ns$ is the state sequence of Eve's
observation $Z^{n\ns}$. The states have joint distribution
$p(s_{ab}^\ns, s_e^\ns)$.  Since the states vary more slowly than the
conditonally-generated sources, there is a length of time $n$ during
which the states remain constant.  This correponds to the sparsity
coherence period discussed earlier.  A large $n$ means that the
states are changing slowly. We assume that the states are available to
the corresponding observers but not to other users. In other words,
Alice and Bob both know $S_{ab}$ but not $S_e$ while Eve knows $S_e$
but not $S_{ab}$.  This is depicted in Figure \ref{fig:sddmms_model}.
We call the state \emph{memoryless} if
\begin{equation}
    p(s_{ab}^\ns, s_e^\ns) = \prod_{m=1}^\ns p(s_{ab,m}, s_{e,m}) \ . \label{eq.stateModel}
\end{equation}
Similarly, the source is \emph{memoryless} if
\begin{align}
    p(x^{n\ns},&y^{n\ns},z^{n\ns}|s_{ab}^\ns, s_e^\ns) \nnum \\
        &= \prod_{m=1}^{\ns} \prod_{i=(m-1)n+1}^{mn} p(x_i,y_i,z_i|s_{ab,m}, s_{e,m}) \ . \label{eq.condGenProcess}
\end{align}
Note that in~\eqref{eq.condGenProcess} one see the effect of the
sparsity coherence period.  The triplet of source samples $(X_i, Y_i,
Z_i)$ is conditionally and independently generated from the same state pair $(S_{ab,m},
S_{e,m})$ for all $i$, $(m-1)n < i \leq mn$.
Each of $(X_i, Y_i, Z_i)$ stands for the vector of channel output in \ref{sec:sparse_channel_model}.

In the one-way discussion protocol (which will be detailed next in
\ref{sec:key_gen_protocol}), Alice sends a message $\commsg$ over a
public channel.  Bob recovers Alice's key based on his observation
$(Y^{n\ns}, S_{ab}^\ns)$ and $\commsg$. Eve's source $Z$ is a \emph{degraded}
version of $Y$ if
\begin{equation}\label{equ:degradness_defn}
    p(x,y,z|s_{ab}, s_{e}) =  p(x,y|s_{ab}) p(z|y, s_{ab}, s_{e}) \ .
\end{equation}
In other words, for given states $(s_{ab}, s_{e})$, Eve's output is a
cascade of the Bob's output and a channel represented by $p(z|y,
s_{ab}, s_{e})$.

\begin{figure}
  \center
  \includegraphics[width=0.55\textwidth]{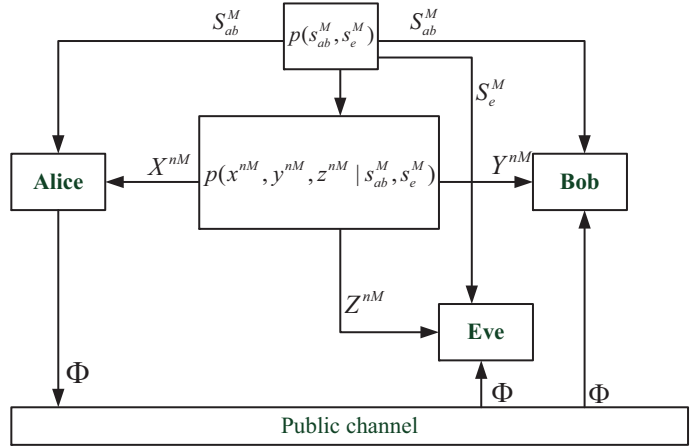}\\
  \caption{State-dependent DMMS model}\label{fig:sddmms_model}
\end{figure}

\subsection{One-way discussion key generation protocol} \label{sec:key_gen_protocol}

Let $\calK=[1:2^{nR}]$ be the key space. There is an authenticated
public channel available to users to exchange error-free public
messages in the set $\Phi=[1:2^{n\Rphi}]$. The one-way public
discussion secret key generation protocol consists of three functions:
\begin{subequations} \label{equ:keygen_system}
\begin{align}
    &f_1  : \calX^{n\ns} \times \calS^\ns \to \calK \ , \\
    &g : \calX^{n\ns} \times \calS^\ns \to \Phi \ , \\
    &f_2  : \calY^{n\ns} \times \calS^\ns \times \Phi \to \calK \ ,
\end{align}
\end{subequations}
which define Alice's key, public message, and Bob's key,
respectively. Namely,
\begin{subequations} \label{equ:oneway_protocol}
\begin{align}
    K &= f_1(X^{n\ns}, S_{ab}^\ns) \ , \\
    \phi &= g(X^{n\ns}, S_{ab}^\ns) \ , \\
    \hatK &= f_2(Y^{n\ns}, S_{ab}^\ns, \phi) \ .
\end{align}
\end{subequations}

\begin{definition}[Achievability]
A secret key rate $R$ is (weakly) achievable if for any $\epsilon>0$,
there is a secret key generation system defined in
\eqref{equ:keygen_system} such that for sufficient large $n$ and
$\ns$,
\begin{align}
  & \Pr(K \neq \hatK) < \epsilon                    \label{equ:reliable_cond} \ , \\
  & \frac{1}{n\ns} I(K; Z^{n\ns}, S_e^\ns, \Phi) < \epsilon   \label{equ:secrecy_cond} \ ,\\
  & \frac{1}{n\ns} H(K) > R - \epsilon              \label{equ:rate_cond} \ .
\end{align}
\end{definition}

Condition \eqref{equ:rate_cond} means the key is almost uniformly
distributed over the set $\calK$. System secrecy is measured in terms
of the mutual information defined in \eqref{equ:secrecy_cond} which
says that the information about the key leaked to eavesdropper is
negligible.  The supremum of achievable secret key rates is called the
\emph{secret key capacity}.

%% file: ergodic01.tex
\section{Ergodic Secret Key Capacity}
\label{sec:sparse.ergodic}

For applications that can tolerate longer delays, the key generation
protocol can operate across a large number of independent state
realizations.  In this setting $n$ and $\ns$ can both be arbitrary
large.  The secret key capacity in the Shannon sense is well-defined
and is termed the \emph{ergodic secret key capacity}, $\Cerg$.

\subsection{Ergodic Capacity of SD-DMMS}

The theorems developed by Ahlswede, \Csiszar \cite[Theorem 1]{Ahlswede_Csiszar93} and Maurer
  \cite[Theorem 1,2]{maurer93} can be applied to the ergodic case of the source
model in Figure \ref{fig:sddmms_model} to get the following lemma.

\begin{lemma}\label{lemma:ergodic_capacity}
  \begin{equation}
    \Cerg^- \leq \Cerg \leq \Cerg^+ \ ,
  \end{equation}
  where
  \begin{align}
    \Cerg^- &= I(X;Y|S_{ab}) - I(X;Z,S_e|S_{ab}) + \frac{1}{n} H(S_{ab}|S_e)  \label{equ:ergodic_LB}\\
    \Cerg^+ &= I(X;Y|Z,S_{ab}, S_e) + \frac{1}{n} H(S_{ab}|S_e) \ . \label{equ:ergodic_UB}
  \end{align}
\end{lemma}
\begin{proof}
The proof is given in Appendix~\ref{sec:proof_ergodic_cap}. $\blacksquare$
\end{proof}

 The important observation about~(\ref{equ:ergodic_LB})
  and~(\ref{equ:ergodic_UB}) is that they both consist of two types of
  terms: mutual information terms and entropy terms.  The latter
  quantifies the amount of uncertainty in the sparsity pattern of the
  main (Alice-to-Bob) channel given Eve's observation $S_e$.  The
  former quantifies the conditional secret key capacity given the
  latter.  The following lemma says that the upper and lower bound
  equal one another when the eavesdropper's observation is degraded.

\begin{corollary}\label{corol:degraded_ergodic_cap}
  For the situation in which the eavesdropper's source is degraded per
  \eqref{equ:degradness_defn}, the ergodic secret key capacity is
  \begin{equation}\label{equ:capacity_sd_dmms}
    \Cerg = I(X;Y|S_{ab}) - I(X;Z,S_e|S_{ab}) + \frac{1}{n} H(S_{ab}|S_e) \ .
  \end{equation}
\end{corollary}

\begin{proof}
  It can be verified by examining \eqref{equ:ergodic_UB} that
  \begin{align*}
    \Cerg^+ &= I(X;Y|Z,S_{ab}, S_e) + \frac{1}{n} H(S_{ab}|S_e) \\
            &= I(X;Y,Z,S_e|S_{ab}) - I(X;Z,S_e|S_{ab}) + \frac{1}{n} H(S_{ab}|S_e) \\
            &= I(X;Y|S_{ab}) - I(X;Z,S_e|S_{ab}) + \frac{1}{n} H(S_{ab}|S_e) \\
            &= \Cerg^- \ ,
  \end{align*}
where the third equality is due to the fact that given $S_{ab}$ we
have Markov chain $X-Y-(Z,S_e)$. This holds since the eavesdropper is degraded.
\end{proof}

Note that when the state changes slowly, which is equivalent to
when $n$ is large, $\frac{1}{n} H(S_{ab}|S_e) \to 0$. That is, the
contribution to the secret key capacity due to the sparsity pattern
$S_{ab}$ is very small. As discussed in
Section~\ref{sec:sparse_channel_model} this will be the common
situation.  Thus, in following, we focus on the non-vanishing term of
\eqref{equ:capacity_sd_dmms}, which we denote as $\Rerg$, \ie,
$\Rerg = I(X;Y|S_{ab}) - I(X;Z,S_e|S_{ab})$.

\subsection{Ergodic secret key rate of sparse wireless channel}
\label{sec:ergodic_wirless}

  We now first apply Lemma~\ref{lemma:ergodic_capacity} to the sparse
  channel model specified in \ref{sec:sparse_channel_model}.  In
  Section~\ref{sec.calcMI} we first examine the expressions for mutual
  information $I(\bmX;\bmY|\SPmain)$ and $I(\bmX;\bmZ,
  \SetSP_e|\SPmain)$ for the vector channel described by
  \eqref{equ:sparse_channel_output} and \eqref{equ:sparse_Eve_output}.
  Then, in Section~\ref{sec.degraded} we identify conditions under
  which the eavesdropper's observation is degraded.  Finally, in
  Sections~\ref{sec.ergodicAch} and~\ref{sec.wideband} we focus in on
  the randomness due to the sparity patterns and analyze the wideband
  limit.

\subsubsection{Mutual information} \label{sec.calcMI}

Define $Q_\ell$ to be the product $S_{ab,\ell} \times S_{e,\ell}$ so
$Q_\ell \in \{0,1\}$. Thus $Q_\ell=1$ if and only if the support (the
sparsity pattern) of $\hmain$ and of $\heve$ are both non-zero in the
$\ell$-th delay bin. Also define two functions:
\begin{subequations} \label{equ:muInfo_formula}
\begin{align}
        &\Iab(\snr_a,\snr_b) = \log \left( \frac{ (1+\snr_a)(1+\snr_b) }{1+ \snr_a + \snr_b} \right) \ , \\
        &\Iae(\snr_a, \snr_e) = \log \left( \frac{ (1+\snr_a)(1+\snr_e) }{1+ \snr_a \snr_e (1-|\eta|^2) +\snr_a + \snr_e} \right) \ . \label{equ:eve_Iae}
\end{align}
\end{subequations}
and $\snr_a = \frac{P}{\nvar_a}$, $\snr_b = \frac{P}{\nvar_b}$ and
$\snr_e = \frac{P}{\nvar_e}$.  We show in Appendix
\ref{sec:sparse.Proof_mutual_info} that
\begin{subequations} \label{equ:vector_muInfo}
\begin{align}
  I(\bmX;\bmY|\SPmain)
    &= E\left[\sum_{\ell=1}^{\maxL} S_{ab,\ell} \Iab(\varh_\ell \snr_a, \varh_\ell \snr_b) \right] \label{equ:Ixy_givenS}\\
  I(\bmX;\bmZ, \SetSP_e|\SPmain)
    &= E\left[ \sum_{\ell=1}^{\maxL} Q_\ell \Iae(\varh_\ell \snr_a, \varhe_\ell \snr_e) \right]\label{equ:Ixz_givenS}
\end{align}
\end{subequations}
In the above expressions, the expectation is taken over the random
sparsity patterns $\SPmain$ and $\SetSP_e$.
Note the factor $Q_\ell$ in
\eqref{equ:Ixz_givenS}.  When $S_{ab, \ell} = 1$ but $S_{e, \ell} = 0$
the eavesdropper has no measurement of that channel coefficient
($Q_{\ell} = 0$).  Thus, the eavesdropper has no observation of that
common randomness and the negative mutual information term in
\eqref{equ:ergodic_LB} is zero.

It is clear in \eqref{equ:vector_muInfo} that channel sparsity
patterns ($\SPmain$ and $\SetSP_e$) effect the mutual information via
the channel DoF (the number of terms in the summation) and the
correlation coefficient $\eta$ effects the information leakage via
$\Iae(\cdot)$ in each delay bin observed by the eavesdropper.

\subsubsection{Degraded condition}\label{sec.degraded}

Because the Eve's channel is correlated to the main channel, she may
get a good estimation of $\hmain$ if she has a higher SNR than Alice
and Bob.  To guarantee the positivity of the secret key rate, we need
to characterize the conditions under which the eavesdropper has a
worse observation than Alice and Bob.  To develop such conditions we
first consider a delay bin
where $Q_\ell=1$\footnote{In subspaces such that $Q_\ell = 0$
either $S_{ab, \ell} = 0$ or Eve has no observation of $H_{ab,\ell}$,
so there is no need to consider those subspaces.}.
Project the channel outputs onto $\bd_\ell$, the $\ell$-th column of
$\bD$, we get
\begin{subequations} \label{equ:scalar_xy}
\begin{align}
   X_\ell &= \bd_\ell^H \bmX \doteq P H_\ell + W_{1,\ell} \\
   Y_\ell &= \bd_\ell^H \bmY \doteq P H_\ell + W_{2,\ell} \ .
\end{align}
\end{subequations}
Because the sounding signal is an (asymptotically) white sequence,
$X_\ell$ (and $Y_\ell$) are sufficient statistic for estimating
$H_\ell$. The noise $W_{1,\ell}$ (resp. $W_{2,\ell}$) is a zero mean
complex Gaussian with variance $P\nvar_a$ (resp. $P\nvar_b$).
Similarly, Eve's sufficient static is
\begin{align}
   Z_\ell &= \bd_\ell^H \bmZ \doteq P H_{e,\ell} + W_{3,\ell} \nnum
   \\ &\equiv P \left( \frac{\upsilon_\ell}{\nu_\ell} \eta H_\ell +
   \sqrt{1 - |\eta|^2} H'_\ell \right) + W_{3,\ell}
   \ . \label{equ:eve_equiv_rx}
\end{align}
Because of $\eta(H_\ell, H_{e,\ell}) = \eta$, we have equivalently
written $H_{e,\ell}$ as a sum of two terms.  The first term is a
scaled version of $H_\ell$.  The second, $H'_\ell$, is a
$\randcn{0}{\varhe_\ell}$ random variable that is independent of
$H_\ell$.  We see from \eqref{equ:eve_equiv_rx} that Eve's observation
$Z_\ell$ contains two types of noise.  The first is the receiver noise
$W_{3,\ell}$.  The second is due to the uncorrelated $H'_\ell$.

Eve's observation $Z_\ell$ will be a degraded version of $Y_\ell$ if
Eve has a smaller SNR than Bob.  This occurs if
\begin{equation}
    \frac{ \varh_\ell P }{\nvar_b} > \frac{|\eta|^2 \varhe_\ell
      P}{(1-|\eta|^2)\varhe_\ell P + \nvar_e} \ .
\end{equation}
Otherwise, $Y_\ell$ is a degraded version of $Z_\ell$. If the sounding
signal power is small and Eve has a suitably smaller noise
variance $\nvar_e$, in particular, when
\begin{equation}
    (1-|\eta|^2) \varhe_\ell P < |\eta|^2
  \frac{\varhe_\ell}{\varh_\ell} \nvar_b - \nvar_e \ ,
\end{equation}
then Bob's output is noisier and no secret key can be extracted at a
positive rate from the $\ell$-th delay bin. This is because when
$P$ is small, Eve's independent noise (due to $H'_\ell$) is decreased.
It is observed in \cite{chou_isit10} that there is a cutoff SNR below
which the secret key capacity is zero.  If $\varh_\ell = \varhe_\ell$
and all the users (Alice, Bob and Eve) are with the same SNR, \ie,
$\nvar_a=\nvar_b=\nvar_e = \nvar$, the secret key capacity will be
positive because Eve has an extra noise (due to the uncorrelated
$H'_\ell$).

\subsubsection{Achievable secret key rate} \label{sec.ergodicAch}

In order to see the effect of channel sparsity when the bandwidth is
large (but finite), we focus on the equal-SNR case and consider a
uniform delay profile, thus, having a degraded eavesdropper.
Define the random number of non-zero channel
coefficients in the main Alice-to-Bob and in Eve's channel to be,
respectively,
\begin{subequations}\label{equ:Num_nonzero_bin}
  \begin{align}
    &B_{ab} = \sum_{\ell=1}^{\maxL}  S_{ab,\ell} \ , \\
    &B_e = \sum_{\ell=1}^{\maxL} S_{e,\ell} \ ,
  \end{align}
\end{subequations}
Note that $B_{ab}$ and $B_e$ are binomial $\randBino{\maxL}{\rho}$
distributed random variables.  Consider a uniform delay profile, \ie,
$\varh_\ell = \frac{1}{B_{ab}}$ for all $\ell$ for which
$S_{ab,\ell}=1$; similarly, $\varhe_\ell = \frac{1}{B_e}$ for all
$\ell$ for which $S_{e,\ell}=1$.

Let $\Iinst(P)$ be the instantaneous key rate $I(\bmX;\bmY|\SPmain) -
I(\bmX;\bmZ, \SetSP_e|\SPmain)$ for fixed $\SPmain$ and $\SetSP_e$. \ie,
\begin{equation} \label{equ:instRate_formula}
  \Iinst(\snr) = B_{ab}  \Iab\left( \frac{\snr}{B_{ab}}, \frac{\snr}{B_{ab}} \right)
                -B_{q} \Iae\left( \frac{\snr}{B_{ab}}, \frac{\snr}{B_e} \right)  \ ,
\end{equation}
where $\snr \triangleq \frac{P}{\nvar}$ and
\begin{equation}\label{equ:Num_overlap_bin}
  B_{q} = \sum_{\ell=1}^{\maxL} Q_{\ell} \quad \mbox{(the number of overlap delay bins)} \ .
\end{equation}
From \eqref{equ:vector_muInfo} and
Corollary \ref{corol:degraded_ergodic_cap} the achievable secret key
rate is
\begin{equation}
  \Ierg(\snr) = E\left[ \Iinst(\snr) \right] \ .
\end{equation}

As we will see later in \ref{sec.wideband}, $\Iinst(\snr)$ is convex in low
SNR (and so is $\Ierg(\snr)$). Thus, a uniform sounding strategy using
a sounding signal with constant power $P$ is not optimal. Let $\calP$
denote all sounding policies that satisfy average power constraint
$E[\bd^H\bd]\leq P$, we can achieve
\begin{equation} \label{equ:defn_Rerg}
  \Rerg(\snr) = \max_{\calP} \Ierg(\snr) \ .
\end{equation}
Note that from Corollary \ref{corol:degraded_ergodic_cap}
and the discussion thereafter, $\Rerg(P)$ approaches $\Cerg(P)$
from the below as $n \to \infty$.

\begin{theorem}[An on-off sounding achieves capacity] \label{thm:sparse_skrate}
    \begin{equation} \label{equ:ergodic_skrate_express}
    \Rerg(\snr) = \max_{0< \tsf \leq 1} \tsf \Ierg\left( \frac{\snr}{\tsf} \right) \ .
    \end{equation}
\end{theorem}
\begin{proof}
The proof is provided in Appendix
\ref{sec:Proof_sparse_skrate}.
\end{proof}

The physical interpretation of the auxiliary variable $\tsf$ is that
key rate $\tsf \Ierg\left( \frac{\snr}{\tsf} \right)$ can be
achieved by an on-off sounding strategy that sounds the channel during
$\tsf$ ($0<\tsf \leq1$) fraction of the time,
each with power $\frac{P}{\tsf}$, and does not sound the channel (i.e., is turned off)
during the rest of the time (\ie, a time-sharing scheme).
Theorem \ref{thm:sparse_skrate} says that the ergodic secret key capacity can be achieved by
an $\tsf^*$ on-off sounding strategy where $\tsf^*$ is the argumment maximizing
\eqref{equ:ergodic_skrate_express}.
As we will discuss in \ref{sec.wideband}, an optimal on-off signal is sparse in time
(\ie, $\tsf^* \to 0$) in a low SNR ($\snr \to 0$) and is dense (\ie, $\tsf^* \to 1$) in a high SNR.

\subsubsection{Wideband regime} \label{sec.wideband}

One way to increase the secret key capacity is to increase the
bandwidth $W$ of the wireless channel.  However, the channel DoF do
not grow linearly in $W$. To see how $W$ effects the secret key rate,
we examine $\Rerg(P)$ from \eqref{equ:ergodic_skrate_express} in the
wideband regime.

In this case, each channel DoF is sounded at a low SNR. At low
SNR we can approximate \eqref{equ:muInfo_formula} as
\begin{subequations}\label{equ:lowSNR_muInfo_approx}
\begin{align}
    &\Iab(x, x)  \approx \frac{x^2}{\ln 2}  \ , \\
    &\Iae(x, y)  \approx \frac{|\eta|^2 xy}{\ln 2}  \ .
\end{align}
\end{subequations}
for $x$ and $y$ small.
The ergodic key rate
\begin{align}
    \Ierg(\snr) &\approx \frac{1}{\ln 2}
        E\left[ B_{ab} \left( \frac{\snr}{B_{ab}} \right)^2 - B_q |\eta|^2 \frac{\snr}{B_{ab}}\frac{\snr}{B_e} \right]  \nnum \\
        &= \frac{\snr^2}{\ln 2} E\left[ \frac{1}{B_{ab}} - |\eta|^2 \frac{B_q}{B_{ab}}\frac{1}{B_e} \right] \nnum \\
        &\stackrel{(a)}{\approx} \frac{\snr^2}{\ln 2} \frac{(1- \theta|\eta|^2)}{L}
        = \frac{\snr^2}{\ln 2} \frac{(1- \theta |\eta|^2)}{(\maxdelay W)^\delta} \ . \label{equ:wideband_approx_ergodic}
\end{align}
The approximation $(a)$ is accurate when $L\gg1$ \cite[eq.(5)]{Grab54}.
The right hand side of \eqref{equ:wideband_approx_ergodic} is a
quadratic  function of $\snr$, thus $\Ierg(\snr)$ is convex in low SNR.

\begin{figure}
  \centering
  \includegraphics[width=0.55\textwidth]{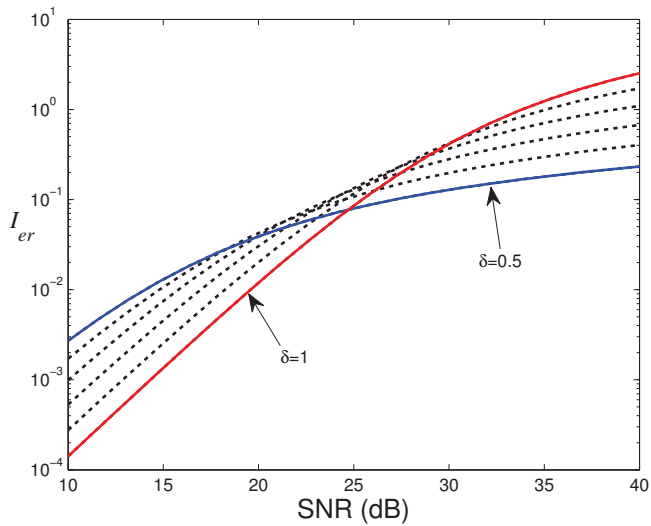}\\
  \caption{Achievable secret key rate $\Ierg(\snr)$ plotted versus SNR
    ($\snr$). The bandwidth is $W=$100MHz, the maximum delay spread
    $\maxdelay =10\mu$s, the conditional probability of overlap in
    $\bmS_{ab}$ and $\bmS_e$ is $\theta=0.5$ and the correlation
  between channel coefficients is $\eta=0.1$. The sparsity parameter
  $\delta \in [0.5,1]$.}
  \label{fig:Ierg_vs_snr}
\end{figure}

Figure \ref{fig:Ierg_vs_snr} plots $\Ierg(\snr)$ versus $\snr$, for a
bandwidth of $W=100$MHz, $\maxdelay=10\mu$s, and for various values of
the sparsity parameter $\delta \in [0.5, 1]$. We see that a sparser
channel (small $\delta$) achieves a higher key rate at low SNRs. We
can also observe this from \eqref{equ:wideband_approx_ergodic}.
According to Theorem \ref{thm:sparse_skrate} and notice that \eqref{equ:wideband_approx_ergodic}
is quadratic in $\snr$, we need a sparser signal in time (an on-off signal) to
get a higher key rate. In other words, in the wideband (power-limited) regime, fewer DoF (either in
channel or in time domain) can achieve a higher key rate.
This occurs because the key generation problem is a combined channel sounding
and channel coding problem.  By focusing energy on fewer DoF we raise their SNR,
enabling key generation to occur at a higher rate.
In contrast, a richer channel (large $\delta$) results in a higher key rate at a high
SNR since that is a DoF-limited (and not a power-limited) regime.

\begin{figure}
  \centering \subfigure[]
             {\includegraphics[width=0.55\textwidth]{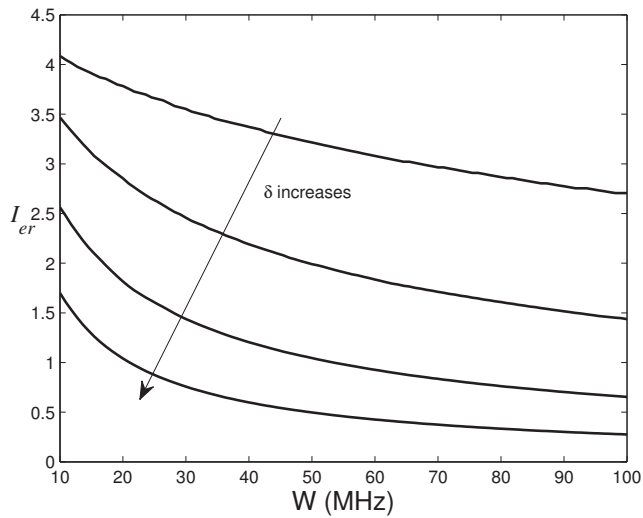} \label{fig:Ierg_vs_w.1}
             } \subfigure[]
             {\includegraphics[width=0.55\textwidth]{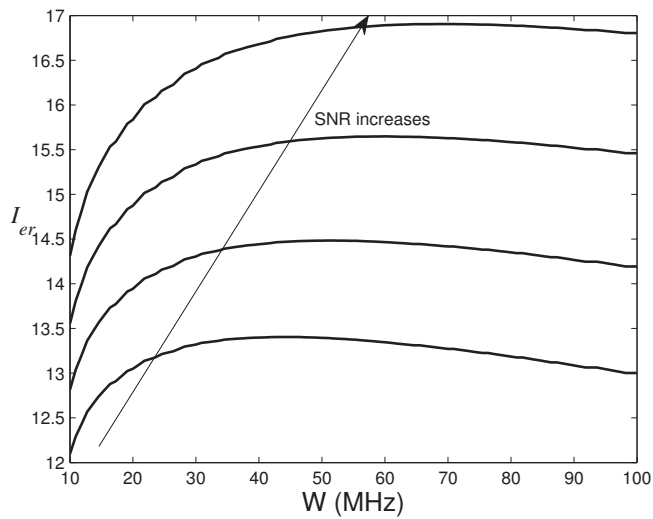} \label{fig:Ierg_vs_w.2}
             }
  \caption{Achievable secret key rate $\Ierg(\snr)$ for SNR fixed at
            $\snr = 10dB$ plotted versus bandwidth $W$.  In
            subfigure \subref{fig:Ierg_vs_w.1} the tradeoff is
            plotted for for values of the sparsity parameter $\delta \in [0.5,1)$.
            In subfigure \subref{fig:Ierg_vs_w.2} the sparsity
            parameter is fixed at $\delta=0.5$ and the tradeoff is plotted for
            four SNRs, $\gamma~(dB) \in [15,16]$.}
  \label{fig:Ierg_vs_w}
\end{figure}

Figure \ref{fig:Ierg_vs_w.1} and \ref{fig:Ierg_vs_w.2} plots $\Ierg$
as a function of $W$.  This provides another view of the tradeoff
between power and DoF. First, let $\snr$ be fixed at $10$dB. Then,
Figure \ref{fig:Ierg_vs_w.1} plots $\Ierg$ for different values of
channel sparsity $\delta$ in the range $[0.5,1)$. In the wideband
(low-SNR) regime, larger $\delta$ results in a smaller key rate.  In
Figure \ref{fig:Ierg_vs_w.2}, $\delta=0.5$ is fixed and $\snr$ is
varied from 10dB to 30dB. We see that for each SNR, there is a unique
$W^*$ that achieves the highest key rate.

%% file: outage01.tex
\section{Secrecy Outage}
\label{sec:sparse.outage}

In contrast to Section~\ref{sec:sparse.ergodic} when an application
has a stringent delay requirement or when the state (i.e., the
sparsity pattern) changes so slowly that it is roughly constant during
the secret key generation process, the secret key capacity in th
Shannon sense is not well-defined.  To study this setting, in this
section we set $\ns=1$ while allowing $n$ to be arbitrary large.
Since users only know their state but not Eve's state, they cannot
adapt the key generation rate to Eve's state. Thus, satisfying the
secrecy condition \eqref{equ:secrecy_cond} can be problematic.  In
this section, we consider an ``outage'' setting with a degraded
eavesdropper.  For any ($\ns=1$) realization $(S_{ab},S_e) =
(s_{ab},s_e)$, we say that a secrecy outage occurs if
\begin{equation}\label{equ:defn_outage_event}
    \frac{1}{n} I(K; Z^{n}, \Phi| s_{ab}, s_e) > \Reve
\end{equation}
for some $\Reve>0$. Namely, there is a non-vanishing information rate
leaked to Eve.  Let
\begin{equation}\label{equ:instRate_given_state}
    \Cs(s_{ab},s_e) = I(X;Y|s_{ab}) - I(X;Z|s_{ab},s_e)
\end{equation}
be the conditional secret key capacity for state $(s_{ab},s_e)$.
Theorem \ref{thm:info_leak} shows that the event \\
$R > \Cs(s_{ab},s_e)$
is a necessary and sufficient condition for the outage event
\eqref{equ:defn_outage_event}.

\begin{theorem} \label{thm:info_leak}
    For any rate-$R$ secret key generation systems for which
    Bob can reliability \\
    recover $X^n$ (\ie, $\Pr(X^n \neq f_2(Y^n,
    s_{ab}, \phi)) \to 0$ for some $f_2(\cdot)$), and let \\
    $\Reve(s_{ab},s_e) = R - \Cs(s_{ab},s_e) > 0$, then
    \begin{enumerate}
      \item[(i)] the information leaked to Eve is lower bounded as
        \begin{equation} \label{equ:info_leak_LB}
            \frac{1}{n} I(K; Z^{n}, \Phi| s_{ab}, s_e) \geq \Reve(s_{ab},s_e) -
            2\epsilon \ ,
        \end{equation}
\end{enumerate}
\quad and
\begin{enumerate}
      \item[(ii)] there exist a coding scheme (\cf,
        \eqref{equ:oneway_protocol}) that satisfies
        \eqref{equ:reliable_cond}, \eqref{equ:rate_cond} and
        \begin{equation} \label{equ:info_leak_UB}
            \frac{1}{n} I(K; Z^{n}, \Phi| s_{ab}, s_e) \leq \Reve(s_{ab},s_e) + 2\epsilon \ .
        \end{equation}
    \end{enumerate}
\end{theorem}
\begin{proof}
The proof can be found in Appendix \ref{sec:sparse.Proof_info_leakage}.
\end{proof}

\subsection{Wideband sparse channel}

In the reciprocal wireless channel case, we know from
\ref{sec:ergodic_wirless} that $\Cs(S_{ab},S_e) = \Iinst(\snr)$ given in
\eqref{equ:instRate_formula}.  Using the approximations from
\eqref{equ:lowSNR_muInfo_approx}, we have in the wideband regime that
\begin{equation}
\label{equ:inst_skcap}
    \Cs(\SPmain,\SetSP_e) \approx \frac{\snr^2}{\ln 2} \left(
    \frac{1}{B_{ab}} - |\eta|^2 \frac{B_q}{B_{ab}}\frac{1}{B_e} \right) \ .
\end{equation}
Recall that $B_{ab}$ (resp. $B_e$) is the weight of vector $\SPmain$
(resp. $\SetSP_e$) (\cf, \eqref{equ:Num_nonzero_bin}) and $B_q$
is the weight of support common to $\SPmain$ and $\SetSP_e$ (\cf,
\eqref{equ:Num_overlap_bin}).  Also note that $B_{ab}$, $B_e$, $B_q$
are random variables so that the overall quantity in
\eqref{equ:inst_skcap} is a random variable.  Unfortunately, there is
no simple expression for the distribution of $\Cs(\SPmain,\SetSP_e)$
in \eqref{equ:inst_skcap}. Since the users are assumed to know
$\SPmain$ (and, therefore, $B_{ab}$), in order to understand how the
channel sparsity effects the probability of outage, we consider the
case where $B_{ab} = L$ (\ie, its mean which is its most likely value)
and $B_{e} = L$.  The only uncertainty at users is $B_q$, the number
of delay bins from which Eve can learn the key.
Conditioned on $B_{ab} = L$ and according to the correlation model
in \eqref{equ:overlap_model},
$B_q$ has a Binomial distribution $\randBino{L}{\theta}$.

From Theorem \ref{thm:info_leak}, the outage probability is
\begin{align}
    \Pout &= \Pr\left( R > \Cs(\SPmain,\SetSP_e) \right) \nnum \\
          &\approx \Pr\left( R > \frac{\snr^2}{L \ln 2} \left( 1 - |\eta|^2 \frac{B_q}{L} \right) \right) \nnum \\
          &= \Pr\left( B_q > \frac{1}{|\eta|^2} \left( 1- \ln 2 \frac{LR}{\snr^2} \right) \right)\ , \label{equ:pout_wideband}
\end{align}
where $L = (W \maxdelay)^{\delta}$. We see that a larger $\snr$ (SNR),
a larger $W$, or a smaller $\eta$ will decrease the outage
probability. However, the sparsity $\delta$ changes both the
distribution of $B_q$ and the right hand side of the argument in
\eqref{equ:pout_wideband} via $L$. Thus, it is still not clear how
$\delta$ impact $\Pout$. When users don't know the instantaneous
secret key capacity, a conservative strategy is to generate a key at a
smaller rate.

Consider a strategy in which the key is generated at rate $R = \alpha
\Ierg(\snr)$ $(0<\alpha<1)$, \ie, a backoff from the ergodic key rate
\eqref{equ:wideband_approx_ergodic}.  We refer to this strategy as the
``$\alpha$-backoff'' strategy. The outage probability
\eqref{equ:pout_wideband} can now be simplified to be
\begin{align}
    &\Pout \approx \Pr\left( B_q \geq a L \right) \label{equ:outage_rearranged}\\
    \text{where} \quad
    & a= (1-\alpha)\frac{1}{|\eta|^2} + \alpha \theta \ . \nnum
\end{align}

Since $B_q$ approximates $\randBino{L}{\theta}$,
the sparsity $\delta$ (and therefore the actual DoF
$L$) determines the distribution of Eve's DoF $B_q$ to observe the main channel.
From \eqref{equ:outage_rearranged}, we can see that when the
$\alpha$-backoff strategy is used,
the correlation coefficient $\eta$ determines how fast the
  threshold $aL$ deviates from
  $\theta L$ (the mean of $B_q$) as $\alpha$ decreases.  Note that in
  \eqref{equ:outage_rearranged}, the SNR (or equivalently, power $P$)
  does not appear in the formula. This is because the key rate is
  proportional to ergodic key rate $\Ierg(\snr)$, which is a quadratic
  function of $\snr$ in the wideband regime (cf.
  \eqref{equ:wideband_approx_ergodic}), and thus cancels the $\snr^2$
  in \eqref{equ:pout_wideband}.

We next use results from large deviation theory \cite{Arratia89} to
upper bound the tail probability of binomial distribution.
\begin{lemma}\cite[Theorem 1]{Arratia89}
   Let $S_n$ be a binomial random variable $\randBino{n}{p}$. For
   $p< a <1$, \\
   and for $n=1,2,3,\cdots$, then
   \begin{equation}
      \Pr(S_n \geq a n) \leq 2^{-n D(a\|p)}
   \end{equation}
   where
   \begin{equation}
      D(a\|p) \equiv a \log \frac{a}{p} + (1-a)\log \frac{(1-a)}{1-p}
   \end{equation}
   is the Kullback-Leibler divergence between the probability
   distributions $\randBern{a}$ and $\randBern{p}$.
\end{lemma}
By this lemma, the outage probability is upper bounded as
\begin{equation} \label{equ:outage_UB}
    \Pout \leq 2^{-L D(a\|\theta)} \ .
\end{equation}

\begin{figure}
  \centering
  \includegraphics[width=0.55\textwidth]{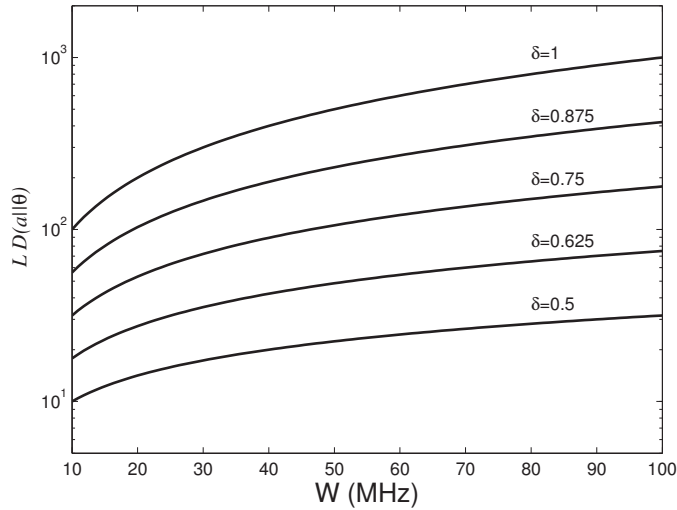}\\
  \caption{Plot of outage exponent $L D(a\|\theta)$ vs. bandwidth
    $W$. Other parameters include the maximum delay spread $\taum =
    100$ns, the conditional probability of overlap in $\bmS_{ab}$ and
    $\bmS_e$ is $\theta=0.5$ and the correlation between channel
    coefficients is $\eta=0.1$.  The sparsity parameter $\delta$ is
    plotted for various values between $0.5 \leq \delta \leq 1$.}
  \label{fig:outageExp_vs_w}
\end{figure}

 Figure \ref{fig:outageExp_vs_w} plots the numerical results of the
 secrecy outage exponent $L D(a\|\theta)$ in the wideband regime.  It
 shows that when the $\alpha$-backoff strategy is used, the mechanism
 through which the channel sparsity impacts the outage probability
 differs from how the channel sparsity impacts the ergodic secret key
 rate.  A richer channel (larger $\delta$) always has larger exponent
 than a sparser channel.  In contrast, Figure \ref{fig:Ierg_vs_snr}
 demonstrates a sparser channel yields a higher ergodic secret key
 rate in the wideband regime.

%% file: conclusion.tex
\section{Conclusions}
\label{sec:sparse.conclusion}

In this paper we study a setting in which two users desire to distill
a common secret key based on the inherent randomness of a reciprocal
wireless channel.  Our particular interest is the effect of channel
sparsity (e.g., in delay), which scales sub-linearly in signal
bandwidth, on secret key generation.  Channel sparsity affects the
inherent randomness of the main channel and increases eavesdropper's
observability of the main channel.  Since channel sparsity is an
  important characteristic of many real-world wireless channels and
  since it has such a large impact on secret key capacity, it is
  crucial to understand this interplay fully.  This will help us to
  deliver secure communication systems with robust guarantees.

We first consider the ergodic setting.  In this setting, at each SNR
there is an adequate bandwidth to maximizes the ergodic secret key
rate. Moreover, when a white sounding sequence is adopted in the
wideband (low-SNR) regime, a higher secret key rate can be achieved by
a sparser channel.

For channels whose sparsity changes relatively slowly, a secrecy
outage measure of performance is adopted.  If the key rate is a
fraction $\alpha$ of the ergodic capacity, we show that richer
channels always have larger exponents characterizing the decay of the
outage probability.  This result illustrates that a larger number of
DoF can smooth out of the detrimental effects of an unknown
eavesdropper state.

%% file: appendix.tex

\subsection{Proof of Lemma \ref{lemma:ergodic_capacity}} \label{sec:proof_ergodic_cap}
\input{Proof_ergodic_cap}

\subsection{Derivation of mutual information \eqref{equ:vector_muInfo}} \label{sec:sparse.Proof_mutual_info}
\input{Proof_mutual_info2}

\subsection{Proof of Theorem \ref{thm:sparse_skrate}}  \label{sec:Proof_sparse_skrate}
\input{Proof_sparse_skrate}

\subsection{Proof of Theorem \ref{thm:info_leak} } \label{sec:sparse.Proof_info_leakage}
\input{Proof_info_leakage}

%% file: Proof_ergodic_cap.tex
Apply the results in \cite{Ahlswede_Csiszar93}, \cite{maurer93} where,
respectively, $(X^{n}, S_{ab})$ are Alice's, $(Y^{n},
S_{ab})$ are Bob's, and $(Z^{n}, S_e)$ are Eve's
observations. A lower bound on the ergodic capacity is 
(see \cite[Theorem 1]{Ahlswede_Csiszar93}, \cite[Theorem 3]{maurer93})
\begin{align*}
\Cerg^- &= \frac{1}{n}  \Big[ I(X^{n}, S_{ab}; Y^{n}, S_{ab} )  - I(X^{n}, S_{ab}; Z^{n}, S_e) \Big].
\end{align*}
Using the chain rule and memoryless nature of the source model
in~(\ref{eq.stateModel}) and~(\ref{eq.condGenProcess}), the first term
can be reduced to
\begin{align*}
&I(X^{n}, S_{ab}; Y^{n}, S_{ab} ) \\
&= I(S_{ab}; Y^{n}, S_{ab}) + I(X^{n}; Y^{n}, S_{ab}|S_{ab}) \\
&= H(S_{ab}) + \sum_{i=1}^{n} I(X_i; Y^{n}|X^{i-1}, S_{ab}) \\
&\eqa H(S_{ab}) + \sum_{i=1}^{n} I(X_i; Y_i|S_{ab}) \\
&= H(S_{ab}) + n I(X; Y|S_{ab}) \ ,
\end{align*}
where $(a)$ follows by applying the memoryless property to source $(X_i, Y_i)$.
Similarly, the second term reduces to
\begin{align*}
&I(X^{n}, S_{ab}; Z^{n}, S_e) \\
&= I(S_{ab}; Z^{n}, S_e) + I(X^{n}; Z^{n}, S_e|S_{ab}) \\
&\eqb I(S_{ab}; S_e) + \sum_{i=1}^{n} I(X_i; Z^{n}, S_e | X^{i-1}, S_{ab}) \\
&= I(S_{ab}; S_e) + \sum_{i=1}^{n} I(X_i; Z_i, S_{e}|S_{ab}) \\
&= I(S_{ab}; S_e) + n I(X; Z, S_e |S_{ab}) \ ,
\end{align*}
where $(b)$ is due to the Markov chain $S_{ab} - S_e - Z^{n}$ 
and applying chain rule on the second term.  Thus, the lower bound is
\begin{equation}
   \Cerg^- = I(X;Y|S_{ab}) - I(X;Z,S_e|S_{ab}) + \frac{1}{n} H(S_{ab}|S_e) \ .
\end{equation}

The upper bound is the conditional mutual information 
(see \cite[Theorem 1]{Ahlswede_Csiszar93}, \cite[Corollary 1]{maurer93}):
\begin{align*}
    \Cerg^+
      &=\frac{1}{n} I(X^{n}, S_{ab}; Y^{n}, S_{ab} | Z^{n}, S_e) \\
      &= \frac{1}{n} \Big( I(S_{ab}; Y^{n}, S_{ab} | Z^{n}, S_e) \\
          & \hspace{1in} + I(X^{n}; Y^{n}, S_{ab} | Z^{n}, S_e, S_{ab}) \Big) \\
      &\leq \frac{1}{n} \Big( H(S_{ab}| Z^{n}, S_e) \\
          & \hspace{0.8in} + \sum_{i=1}^{n} I(X_i; Y^{n} | X^{i-1}, Z^{n}, S_e, S_{ab}) \Big) \\
      &\eqc \frac{1}{n} \Big( H(S_{ab}| S_e)  +  \sum_{i=1}^{n} I(X_i; Y_i |  Z_i, S_{e}, S_{ab}) \Big) \\
      &= I(X;Y|Z,S_{ab}, S_e) + \frac{1}{n} H(S_{ab}|S_e) \ ,
\end{align*}
where $(c)$ follows by Markov condition $S_{ab} - S_e - Z^{n}$ and the memoryless property.

%% file: Proof_mutual_info2.tex
First consider $I(\bmX;\bmY|\SPmain)$:
\begin{align}
  I(\bmX;\bmY | \SPmain ) &= E[ h(\bmX|\SPmain) - h(\bmY|\SPmain) -
        h(\bmX, \bmY|\SPmain) ] \nnum \\
        &= E \left[ \log \left( \frac{\det(\bR_{\bmX})\cdot \det(\bR_{\bmY})}{\det(\bR_{\bmX\bmY})} \right) \right], \label{equ:det_ratio}
\end{align}
where $h(\bmX)$ is the differential entropy \cite{CoverIT} of $\bmX$,
and the expectation is taken over the distribution of $\SPmain$. Let
$\bR_{\bmX}$ denote the covariance matrix of $\bmX$ when the input
$\SPmain = \bS$ is fixed, \ie,
\begin{align*}
  \bR_{\bmX} = E[\bmX\bmX^H|\bS] = \bD \bR_{\bh} \bD^H + \nvar_1 \bI_{K} \ ,
\end{align*}
where $\bR_{\bh} = \diag(\varh_1, \cdots, \varh_L)$ and $\varh_\ell=0$ if $S_\ell=0$.
Similarly,
\begin{align*}
  \bR_{\bmY} &=\bD \bR_{\bh} \bD^H + \nvar_2 \bI_{N} \\
  \bR_{\bmX\bmY} &=
    \left[
      \begin{array}{c|c}
        \bR_{\bmX} & \bD \bR_{\bh} \bD^H \\
        \hline
        \bD \bR_{\bh} \bD^H & \bR_{\bmY} \\
      \end{array}
    \right] \ .
\end{align*}

We simplify the determinants as follows,
\begin{align*}
  \det(\bR_{\bmX}) &= \det(\bD \bR_{\bh} \bD^H + \nvar_1 \bI_{N}) \\ &=
        (\nvar_1)^N \det\left(I_N
        + \frac{\bD \bR_{\bh} \bD^H}{\nvar_1} \right) \\ &\eqa
        (\nvar_1)^N \det\left(I_L
        + \frac{ \Lambda \bD^H \bD \Lambda}{\nvar_1} \right) \\ &\eqb
        (\nvar_1)^N \prod_{\ell=1}^L \left(1+ \frac{P}{\nvar_1} \varh_\ell \right) \\
        &=
        (\nvar_1)^N \prod_{\ell:S_\ell=1} \left(1+ \frac{P}{\nvar_1} \varh_\ell \right),
\end{align*}
where $(a)$ follows by defining $\Lambda = \sqrt{\bR_{\bh}}$ and
applying Sylvester's determinant formula: $\det(\bI_m + \bA\bB)
= \det(\bI_n + \bB\bA)$ where $\bA$ is an $m$-by-$n$ matrix, and $\bB$
is an $n$-by-$m$ matrix.  Step $(b)$ is due
to \eqref{equ:white_sequence_assumption}.  Similarly, we find that
\begin{align*}
  \det(\bR_{\bmX}) &\doteq
  (\nvar_2)^N \prod_{\ell:S_\ell=1} \left(1+ \frac{P}{\nvar_2} \varh_\ell \right) \\ \det(\bR_{\bmX\bmY})
  &\doteq
  (\nvar_1 \nvar_2)^N \prod_{\ell:S_\ell=1} \left(1+ \frac{(\nvar_1+\nvar_2)P}{\nvar_1\nvar_2} \varh_\ell \right).
\end{align*}
Substituting into \eqref{equ:det_ratio}, we get \eqref{equ:Ixy_givenS}.

Follow a similar calculation, we get $I(\bmX;\bmZ, \SetSP_e|\SPmain)$
in \eqref{equ:Ixz_givenS} by noting that
\begin{align*}
  \bR_{\bmX\bmZ} &=
    \left[
      \begin{array}{c|c}
        \bR_{\bmX} & \bD \bR_{\bh\tilbh} \bD^H \\
        \hline
        \bD \bR_{\bh\tilbh} \bD^H & \bR_{\bmZ} \\
      \end{array}
    \right] \ .
\end{align*}
where $\bR_{\bh\tilbh}$ is a diagonal matrix and its $\ell$-th
diagonal element is equal to $\eta \varh_\ell$ if $Q_\ell=1$ or is
equal to zero if $Q_\ell=0$.

%% file: Proof_sparse_skrate.tex
The proof is similar to \cite[Theorem 4]{chou_it10}.  First note that
$\Iinst(\snr)$ is non-decreasing in $\snr$ and so is
$\Ierg(\snr)$. This can be verified by evaluating
$\frac{\partial\Iinst(\snr)}{\partial \snr}$, which is non-negative.
Define $\bar{I}_{\mathrm{er}}(\snr) = \max_{0<\tsf\leq
1}\tsf \Ierg\left(\frac{\snr}{\tsf}\right)$. Note that
$\bar{I}_{\mathrm{er}}(\snr)$ is a concave and non-decreasing function
of $\snr$.  We are going to show
$\bar{I}_{\mathrm{er}}\left( \frac{P}{\nvar} \right)$ is equal to
$\Rerg\left( \frac{P}{\nvar} \right)$ defined in \eqref{equ:defn_Rerg}
over the average power constraint $P$.  Let $\calP$ be the set of all
sounding policies satisfying average power constraint $P$.
Specifically, let the sounding policy in $\calP$ allocate power $P_s$
to sounding signals with probability $p(s)$ such that $E[P_s]
= \sum_{s} p(s)P_s \leq P$. Note that
$\Rerg\left( \frac{P}{\nvar} \right) \geq \bar{I}_{\mathrm{er}}\left( \frac{P}{\nvar} \right)$. We
can also upper bound
\begin{align*}
    \Rerg\left( \frac{P}{\nvar} \right)
        &=  \max_{\calP} \sum_{s} p(s) \Ierg\left( \frac{P_s}{\nvar} \right) \\ 
        &\leq  \max_{\calP}  \sum_{s} p(s) \left[ \max_{0<\tsf\leq 1}\tsf \Ierg\left(\frac{P_s}{\tsf \nvar} \right) \right]  \\
        & =  \max_{\calP}  \sum_{s} p(s) \bar{I}_{\mathrm{er}}\left(\frac{P_s}{\nvar} \right) \\
        &\lea \max_{\calP}  \bar{I}_{\mathrm{er}}\left( \frac{ \sum_{s} p(s) P_s}{\nvar} \right) \\
        &\leb  \bar{I}_{\mathrm{er}}\left( \frac{P}{\nvar} \right) \ ,
\end{align*}
where $(a)$ and $(b)$ are due to the concavity and non-decreasing
function of $\bar{I}_{\mathrm{er}}(\snr)$.

%% file: Proof_info_leakage.tex
We first show the lower bound and then the upper bound.

\subsubsection{Lower bound \eqref{equ:info_leak_LB}}

\begin{align*}
  I(&K; Z^n, \Phi | s_{ab}, s_e) \\
    &=  H(K|s_{ab}, s_e) - H(K | Z^n, \Phi, s_{ab}, s_e) \\
    &= H(K|s_{ab}, s_e) - [H(K, \Phi | Z^n, s_{ab}, s_e) - H(\Phi | Z^n, s_{ab}, s_e)] \\
    &\gea H(K|s_{ab}, s_e) + H(\Phi | Y^n, s_{ab}, s_e) \\
    & \quad  - [ H(X^n, K, \Phi | Z^n, s_{ab}, s_e) - H(X^n |K, \Phi, Z^n, s_{ab}, s_e) ] \\
    &\geb n(R - \epsilon) + H(\Phi| Y^n, s_{ab}, s_e) - H(X^n| Z^n, s_{ab}, s_e) \\
    &=  n(R - \epsilon) + H(\Phi| Y^n, s_{ab}, s_e) -  H(X^n| Y^n, s_{ab}, s_e) \\
    & \quad + H(X^n| Y^n, s_{ab}, s_e) - H(X^n| Z^n, s_{ab}, s_e) \\
    &\eqc n(R - \epsilon) - H(X^n| \Phi, Y^n, s_{ab}, s_e) - n \Cs(s_{ab},s_e)\\
    &\ged n(R - \Cs(s_{ab},s_e) - 2\epsilon) \ ,
\end{align*}
where $(a)$ is due to the fact that given $(s_{ab},s_e)$, $\Phi - X^n
- Y^n - Z^n$ form a Markov chain. Thus $H(\Phi | Z^n, s_{ab},
s_e) \geq H(\Phi | Y^n, s_{ab}, s_e)$.  $(b)$ holds because entropy is
non-negative and $(K,\Phi)$ is function of $(X^n, s_{ab})$, we can
take $K, \Phi$ away from $H(X^n, K, \Phi | Z^n, s_{ab}, s_e)$.  For
the same reason, we can add $\Phi$ in $H(X^n| Y^n, s_{ab}, s_e)$ and
use chain rule to get $(c)$.  $(d)$ is due to the reliable condition
$\Pr(X^n \neq f_2(Y^n, s_{ab}, \Phi)) \to 0$ by applying Fano's
inequality \cite{CoverIT}.

\subsubsection{Upper bound \eqref{equ:info_leak_UB}}

\begin{align*}
  I(&K; Z^n, \Phi | s_{ab}, s_e) \\
    &=  H(K|s_{ab}, s_e) - H(K | Z^n, \Phi, s_{ab}, s_e) \\
    &\lea nR - H(K | Z^n, \Phi, s_{ab}, s_e)
\end{align*}
because $K \in [1:2^{nR}]$. We need to show there exist a coding scheme such that
\begin{equation} \label{equ:eve_entropy_LB}
    \frac{1}{n}H(K | Z^n, \Phi, s_{ab}, s_e) \geq \Cs(s_{ab},s_e) - 2\epsilon \ .
\end{equation}
We will use the following lemma in the proof.

\begin{lemma} (cf. \cite[eq.(25)]{CsiszarKorner78} \cite[eq.(16)]{chou_allerton11})
\label{lemma:exist_scheme}
    Any $\epsilon > 0$, if $\frac{1}{n} H(\Phi|s_{ab},s_e) \geq H(X|Y,s_{ab}) + \epsilon$,
    there exists a coding scheme where $f_2(\cdot)$ and $f_3(\cdot)$
    are the decoding functions at Bob and Eve, respectively, such
    that for sufficiently large $n$,
    \begin{enumerate}[(i)]
      \item $\Pr(X^n \neq f_2(Y^n, s_{ab}, \Phi)) \to 0$  \ ,
      \item $\frac{1}{n} H(X^n| K,\Phi, Z^n, s_{ab}, s_e) \leq \epsilon$ \ .
    \end{enumerate}
\end{lemma}
The proof of Lemma \ref{lemma:exist_scheme} uses a random coding
technique to show existence. The first statement is exactly the
Slepian-Wolf theorem \cite{SlepianWolf73}. The second statement is
the standard equivocation analysis which says if Eve knows $K$ and
$\Phi$ along with her observation $Z^n$, she can recover sequence
$X^n$. We refer to \cite{CsiszarKorner78, chou_allerton11} for the
details.

Adopting the coding scheme in Lemma~\ref{lemma:exist_scheme} 
where the public message $\Phi \in [1: 2^{n(H(X|Y,s_{ab}) + \epsilon)}]$.
We prove \eqref{equ:eve_entropy_LB} through a sequence of
(in)equalities:
\begin{align*}
    H(&K | \Phi, Z^n, s_{ab}, s_e) \\
    &=  H(X^n, K| \Phi, Z^n, s_{ab}, s_e) - H(X^n| K, \Phi, Z^n, s_{ab}, s_e) \\
    &\gea H(X^n, K| \Phi, Z^n, s_{ab}, s_e) - n \epsilon \\
    &= H(X^n, K, \Phi| Z^n, s_{ab}, s_e) - H(\Phi| Z^n, s_{ab}, s_e) - n \epsilon \\
    &\eqb H(X^n| Z^n, s_{ab}, s_e) - H(\Phi| Z^n, s_{ab}, s_e) - n \epsilon \\
    &\geq H(X^n| Z^n, s_{ab}, s_e) - H(\Phi) - n \epsilon \\
    &\gec H(X^n| Z^n, s_{ab}, s_e) - n H(X| Y, s_{ab})  - 2n \epsilon \\
    &= n (H(X| Z, s_{ab}, s_e) - H(X| Y, s_{ab})  - 2\epsilon) \\
    &= n ( \Cs(s_{ab},s_e) - 2\epsilon) \ ,
\end{align*}
where $(a)$ is due to  $(ii)$ of Lemma~\ref{lemma:exist_scheme}.
$(b)$ follows because $(K, \Phi)$ is a function of $(X^n, s_{ab})$.
$(c)$ comes from the fact that $\Phi \in [1: 2^{n(H(X|Y,s_{ab}) + \epsilon)}]$
and the entropy is upper bounded by a uniform distribution.
This completes the proof.

%% file: main.bbl
\begin{thebibliography}{10}

\bibitem{Ahlswede_Csiszar93}
R.~Ahlswede and I.~Csisz\'{a}r, ``Common randomness in information theory and
  cryptography part {I}: Secret sharing,'' {\em Information Theory, IEEE
  Transactions on}, vol.~39, no.~4, pp.~1121--1132, 1993.

\bibitem{maurer93}
U.~M. Maurer, ``Secret key agreement by public discussion from common
  information,'' {\em Information Theory, IEEE Transactions on}, vol.~39,
  no.~3, pp.~733--742, 1993.

\bibitem{Khisti_isit08}
A.~Khisti, S.~Diggavi, and G.~Wornell, ``Secret-key generation with correlated
  sources and noisy channels,'' in {\em Proc. Int. Symp. Inform. Theory},
  pp.~1005--1009, July 2008.

\bibitem{Prabhakaran08}
V.~Prabhakaran, K.~Eswaran, and K.~Ramchandran, ``Secrecy via sources and
  channels -- a secret key-secret message rate tradeoff region,'' in {\em Proc.
  Int. Symp. Inform. Theory}, pp.~1010--1014, July 2008.

\bibitem{Wyner75}
A.~Wyner, ``The wire-tap channel,'' {\em The Bell Systems Technical Journal},
  vol.~54, pp.~1355--1387, 1975.

\bibitem{ChenVinck06}
Y.~Chen and A.~Han~Vinck, ``Wiretap channel with side information,'' {\em
  Information Theory, IEEE Transactions on}, vol.~54, pp.~395--402, Jan. 2008.

\bibitem{Liu07}
W.~Liu and B.~Chen, ``Wiretap channel with two-sided channel state
  information,'' in {\em Proc. Asilomar Conf. Signals, Systems and Computers,
  2007}, pp.~893--897, Nov. 2007.

\bibitem{Khisti_skey_asym_isit09}
A.~Khisti, S.~Diggavi, and G.~Wornell, ``Secret key agreement using asymmetry
  in channel state knowledge,'' in {\em Proc. Int. Symp. Inform. Theory},
  pp.~2286--2290, 2009.

\bibitem{KhistiEtal_skey_TCSI}
A.~Khisti, S.~Diggavi, and G.~Wornell, ``Secret-key agreement with channel
  state information at the transmitter.'' arXiv:1009.3052.

\bibitem{chou_it10}
T.~Chou, S.~Draper, and A.~Sayeed, ``Key generation using external source
  excitation: {C}apacity, reliability, and secrecy exponent,'' {\em Information
  Theory, IEEE Transactions on}, vol.~58, pp.~2455--2474, Apr. 2012.

\bibitem{chou_allerton11}
T.~Chou, V.~Tan, and S.~Draper, ``The sender-excited secret key agreement
  model: Capacity theorems,'' in {\em Proc. Allerton Conference on
  Communication, Control, and Computing}, 2011.

\bibitem{hessan96}
A.~A. Hassan, W.~E. Stark, J.~E. Hersheyc, and S.~Chennakeshu, ``Cryptographic
  key agreement for mobile radio,'' {\em Digital Signal Processing}, vol.~6,
  no.~4, pp.~207--212, 1996.

\bibitem{wilson_Tse_Scholtz07}
R.~Wilson, D.~Tse, and R.~A. Scholtz, ``Channel identification: Secret sharing
  using reciprocity in ultrawideband channels,'' {\em Information Forensics and
  Security, IEEE Transactions on}, vol.~2, no.~3, pp.~364--375, 2007.

\bibitem{Ye_vtc07}
C.~Ye, A.~Reznik, G.~Sternberg, and Y.~Shah, ``On the secrecy capabilities of
  {ITU} channels,'' pp.~2030--2034, Oct. 2007.

\bibitem{YeEtal}
C.~Ye, S.~Mathur, A.~Reznik, Y.~Shah, W.~Trappe, and N.~Mandayam,
  ``Information-theoretically secret key generation for fading wireless
  channels,'' {\em Information Forensics and Security, IEEE Transactions on},
  vol.~5, pp.~240--254, June 2010.

\bibitem{Wallace09}
M.~A.~J. Jon W.~Wallace, Chan~Chen, ``Key generation exploiting {MIMO} channel
  evolution: Algorithms and theoretical limits,'' March 2009.

\bibitem{Patwari10}
N.~Patwari, J.~Croft, S.~Jana, and S.~K. Kasera, ``High-rate uncorrelated bit
  extraction for shared secret key generation from channel measurements,'' {\em
  Mobile Computing, IEEE Transactions on}, vol.~9, pp.~17--30, Jan. 2010.

\bibitem{Aono05}
T.~Aono, K.~Higuchi, T.~Ohira, B.~Komiyama, and H.~Sasaoka, ``Wireless secret
  key generation exploiting reactance-domain scalar response of multipath
  fading channels,'' {\em Antennas and Propagation, IEEE Transactions on},
  vol.~53, pp.~3776--3784, Nov. 2005.

\bibitem{sayeed08}
A.~Sayeed and A.~Perrig, ``Secure wireless communications: Secret keys through
  multipath,'' in {\em ICASSP 2008. IEEE International Conference on},
  pp.~3013--3016, 2008.

\bibitem{Molisch05}
A.~Molisch, ``Ultrawideband propagation channels-theory, measurement, and
  modeling,'' {\em Vehicular Technology, IEEE Transactions on}, vol.~54,
  pp.~1528--1545, sept. 2005.

\bibitem{GJRTT04}
S.~Ghassemzadeh, R.~Jana, C.~Rice, W.~Turin, and V.~Tarokh, ``Measurement and
  modeling of an ultra-wide bandwidth indoor channel,'' {\em Communications,
  IEEE Transactions on}, vol.~52, pp.~1786--1796, Oct. 2004.

\bibitem{Yan07}
Z.~Yan, M.~Herdin, A.~M. Sayeed, and E.~Bonek, ``Experimental study of {MIMO}
  channel statistics and capacity via the virtual channel representation,''
  {\em University of Wisconsin-Madison, Tech. Rep.}, Feb. 2007.

\bibitem{Czink07}
N.~Czink, X.~Yin, H.~Ozcelik, M.~Herdin, E.~Bonek, and B.~H. Fleury, ``Cluster
  characteristics in a {MIMO} indoor propagation environment,'' {\em IEEE
  Trans. Wireless Commun.}, vol.~6, pp.~1465--1475, Apr. 2007.

\bibitem{Sayeed06}
A.~Sayeed, ``Sparse multipath wireless channels: Modeling and implications,''
  in {\em Proc. ASAP}, 2006.

\bibitem{Bajwa09}
W.~Bajwa, A.~Sayeed, and R.~Nowak, ``Sparse multipath channels: Modeling and
  estimation,'' in {\em Digital Signal Processing Workshop}, pp.~320--325, Jan.
  2009.

\bibitem{Bajwa10}
W.~Bajwa, J.~Haupt, A.~Sayeed, and R.~Nowak, ``Compressed channel sensing: A
  new approach to estimating sparse multipath channels,'' {\em Proc. IEEE
  (special issue on Sparse Signal Processing)}, June 2010.

\bibitem{LeeTC73}
W.~Lee, ``Effects on correlation between two mobile radio base-station
  antennas,'' {\em Communications, IEEE Transactions on}, vol.~21,
  pp.~1214--1224, Nov. 1973.

\bibitem{Tang09}
X.~Tang, R.~Liu, P.~Spasojevic, and H.~Poor, ``On the throughput of secure
  hybrid-{ARQ} protocols for gaussian block-fading channels,'' {\em Information
  Theory, IEEE Transactions on}, vol.~55, pp.~1575--1591, April 2009.

\bibitem{GungorArXiv}
O.~{Gungor}, J.~{Tan}, C.~{Emre Koksal}, H.~{El Gamal}, and N.~B. {Shroff},
  ``{Secrecy Outage Capacity of Fading Channels},'' {\em arXiv:1112.2791}, Dec.
  2011.

\bibitem{Sayeed07}
A.~Sayeed and V.~Raghavan, ``Maximizing {MIMO} capacity in sparse multipath
  with reconfigurable antenna arrays,'' {\em IEEE Journal on Special Topics in
  Signal Processing (special issue on Adaptive Waveform Design for Agile
  Sensing and Communication)}, pp.~156--166, June 2007.

\bibitem{Cotter02}
S.~Cotter and B.~Rao, ``Sparse channel estimation via matching pursuit with
  application to equalization,'' {\em Communications, IEEE Transactions on},
  vol.~50, pp.~374--377, Mar 2002.

\bibitem{Carbonelli07}
C.~Carbonelli, S.~Vedantam, and U.~Mitra, ``Sparse channel estimation with zero
  tap detection,'' {\em Wireless Communications, IEEE Transactions on}, vol.~6,
  pp.~1743--1763, May 2007.

\bibitem{Li07}
W.~Li and J.~Preisig, ``Estimation of rapidly time-varying sparse channels,''
  {\em Oceanic Engineering, IEEE Journal of}, vol.~32, pp.~927--939, Oct. 2007.

\bibitem{Raghavan07}
V.~Raghavan, G.~Hariharan, and A.~Sayeed, ``Capacity of sparse multipath
  channels in the ultra-wideband regime,'' {\em IEEE Journal on Special Topics
  in Signal Processing (special issue on Performance Limits of Ultra-Wideband
  Systems)}, pp.~156--166, June 2007.

\bibitem{Raghavan2011}
V.~Raghavan and A.~Sayeed, ``Sublinear capacity scaling laws for sparse {MIMO}
  channels,'' {\em Information Theory, IEEE Transactions on}, vol.~57,
  pp.~345--364, Jan. 2011.

\bibitem{sayeed_bookchapter}
A.~M. Sayeed and T.~Sivanadyan, ``Wireless communication and sensing in
  multipath environments using multi-antenna transceivers,'' in {\em Handbook
  on Array Processing and Sensor Networks}, Ch.5, Wiley-IEEE Press, 2010.

\bibitem{SSHandbook}
M.~K. Simon, J.~K. Omura, R.~A. Scholtz, and B.~K. Levitt, {\em Spread Spectrum
  Communications Handbook}.
\newblock McGraw-Hill, 1994.

\bibitem{Proakis_DigitComm}
J.~Proakis, {\em {Digital Communications}}.
\newblock McGraw-Hill, Aug. 2000.

\bibitem{Goldsmith2005}
A.~Goldsmith, {\em Wireless Communications}.
\newblock Cambridge University Press, 2005.

\bibitem{FRG09}
A.~Fletcher, S.~Rangan, and V.~Goyal, ``Necessary and sufficient conditions for
  sparsity pattern recovery,'' {\em Information Theory, IEEE Transactions on},
  vol.~55, pp.~5758--5772, dec. 2009.

\bibitem{chou_isit10}
T.~Chou, S.~Draper, and A.~Sayeed, ``Impact of channel sparsity and correlated
  eavesdropping on secret key generation from multipath channel randomness,''
  in {\em Proc. Int. Symp. Inform. Theory}, 2010.

\bibitem{Grab54}
E.~L. Grab and I.~R. Savage, ``Tables of the expected value of $1/{X}$ for
  positive {B}ernoulli and {P}oisson variables,'' {\em Journal of the American
  Statistical Association}, vol.~49, pp.~169--177, Mar. 1954.

\bibitem{Arratia89}
R.~Arratia and L.~Gordon, ``Tutorial on large deviations for the binomial
  distribution,'' {\em Bulletin of Mathematical Biology}, vol.~51, no.~1,
  pp.~125--131, 1989.

\bibitem{CoverIT}
T.~M. Cover and J.~A. Thomas, {\em Elements of information theory}.
\newblock Wiley series in telecommunications, New York: Wiley, 2nd~ed., 2006.

\bibitem{CsiszarKorner78}
I.~Csisz\'{a}r and J.~K\"{o}rner, ``Broadcast channels with confidential
  messages,'' {\em Information Theory, IEEE Transactions on}, vol.~24, no.~3,
  pp.~339--348, 1978.

\bibitem{SlepianWolf73}
D.~Slepian and J.~Wolf, ``Noiseless coding of correlated information sources,''
  {\em Information Theory, IEEE Transactions on}, vol.~19, pp.~471--480, Jul
  1973.

\end{thebibliography}
